\documentclass[12pt]{amsart}

\pdfoutput=1

\usepackage[text={420pt,660pt},centering]{geometry}

\usepackage{color}
\usepackage{esint,amssymb}
\usepackage{graphicx}
\usepackage{mathtools}
\usepackage[colorlinks=true, pdfstartview=FitV, linkcolor=blue, citecolor=blue, urlcolor=blue,pagebackref=false]{hyperref}
\usepackage{microtype}
\usepackage{bm}
\usepackage{dsfont}
\usepackage[font={footnotesize}]{caption}

\usepackage[notcite,notref,color]{showkeys}

\usepackage{lipsum}

\newtheorem{proposition}{Proposition}
\newtheorem{theorem}[proposition]{Theorem}

\newtheorem{corollary}[proposition]{Corollary}

\theoremstyle{definition}

\newtheorem{problem}[proposition]{Problem}

\numberwithin{equation}{section}
\numberwithin{proposition}{section}

\begin{document}

\title[Obamacare and a Fix for the IRS Iteration]{Obamacare and a Fix for the IRS Iteration}


\author[S. J. Ferguson]{Samuel J. Ferguson}
\address[S. J. Ferguson]{Metron, Inc., Ste 600, 1818 Library St, Reston, VA 20190}
\email{sjf370@nyu.edu}

\maketitle

\section{Introduction}

In 2018, I took an Uber ride. Although my driver qualified for help with paying for health insurance under the Affordable Care Act, he couldn't determine the amount of his benefit. Worse, tax software and government calculators said he should receive \$0 to help him pay for health insurance, instead of the roughly \$3,000 that the law seemed to prescribe. He asked me to look into the matter, and my efforts led to a mathematical odyssey captured by \emph{Time}'s film crew and a senior writer at \emph{Money} magazine in an online film clip and article \cite{TIM}. I am motivated to publish my findings by a communication \cite{OCC} that the IRS will include reference to it in its guidance after publication in a peer-reviewed journal. Then, tax software companies will be able to implement procedures proposed here without legal liability, relieving the current computational issues affecting Affordable Care Act beneficiaries. I am also motivated by the opportunity to bear witness to the resolution of a civic concern by means of mathematical modeling and proof.

\section{Obamacare's Premium Tax Credit}

A tax device created by the Affordable Care Act plays a key role in my driver's problem, so we first review this law. In 2010, the United States Congress passed the Patient Protection and Affordable Care Act \cite{ACA}, also called Obamacare. A couple of its provisions are relevant. First, it provides for the setup of online exchanges, so American households can directly purchase health insurance meeting certain minimum standards. These standards apparently give rise to the ``patient protection'' part of the law's name. Second, the law makes qualified health insurance affordable for every American household with \emph{household income}\footnote{We postpone giving definitions of ``affordable'' and ``household income,'' but we note that the latter may be referred to as the household's modified adjusted gross income (MAGI) in the literature, hence our choice of the letter ``M'' to denote it. Worksheets for calculating $M$ may be found in the Instructions for Form 8962 \cite[p~6]{IF8962}.} $M$ in the range
\[
F\leq M\leq 4F.
\]
Here, $F$ is the \emph{federal poverty line}\footnote{The value of $F$ used in Obamacare calculations for a given tax year may be found in the Instructions for Form 8962 for that year \cite[pp 6--7]{IF8962}. For example, in the continental United States in 2018, for a household with $n$ people, $F$ is approximately $\text{\$8,000}+n\cdot\text{\$4,000}$ \cite[p 7]{IPTC18}. Thus, at that time, a household of one person had a federal poverty line of about \$12,000, and a household of four people had a federal poverty line of about \$24,000.} for the household, a governmentally-prescribed number depending on household size and state which adjusts annually according to a specified notion of inflation.

Now, how does the law make health insurance affordable when $F\leq M\leq 4F$? It does so by creating a tax credit to help eligible households pay the premiums of qualified health insurance. For such households, and for at least some choices of health plan, the credit pays all of the cost of the insurance premiums except for a portion which is considered affordable. Moreover, this credit may be received in advance, to help pay the cost of premiums right away, and is refundable, so the full credit is receivable whether or not the household owes taxes that offset it.

Having introduced Obamacare's tax credit for premiums, how can we find it? Following along with Form 8962 \cite{F8962}, which taxpayers must file to claim the premium tax credit, we see that the computation requires an \emph{applicable figure} $f$. The applicable figure represents the percentage of $M$ which is affordable for the household to pay for health insurance. Working with decimals rather than percentages, we model $f$ as a governmentally-determined function of $m$, where
\[
m=M/F.
\]
Here, the division is exact, so $m$ is a real number in $[1,4]$ obtained without rounding when $F\leq M\leq 4F$. If $M>4F$, then household income is too high to receive the premium tax credit (PTC), so $PTC=\$0$. When $M<F$ we again have $PTC=\$0$ unless the household qualifies for an exception \cite[p 8]{IF8962}, in which case the applicable figure $f(1)$ is used. Thus, applicable figures aren't needed for $m$ outside $[1,4]$.

\begin{figure}
  \includegraphics[width=\linewidth]{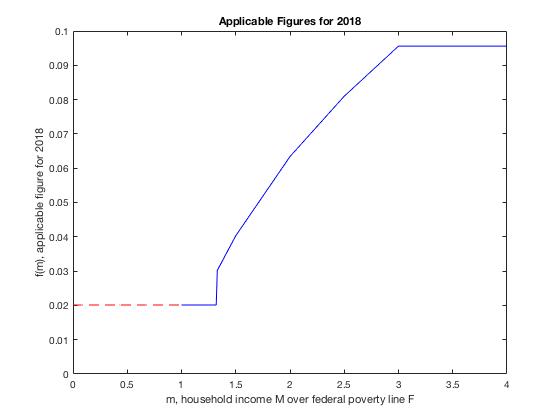}
  \caption{A graph of the 2018 applicable figure $f(m)$ as a function of $m$. For $m<1$, the value of the dashed line is $f(1)$. This value is used if an applicable figure is required in that case. No applicable figure is needed when $m>4$.}
  \label{fig:ApplicableFiguresGraph}
\end{figure}
\emph{Example.} Say we are considering the 2018 tax year. Then the applicable figure is appropriately modeled\footnote{The values $j=f(1)$, $k=f(1.33)$, $\ell=f(1.5)$, $a=f(2)$, $b=f(2.5)$, $c=f(3)$ are from Table 2 in the Instructions for Form 8962 for 2018 \cite[p 9]{IPTC18}. They are found by locating applicable figures for household income equal to 100\%, 133\%, 150\%, 200\%, 250\%, and 300\% of the federal poverty line, respectively. That our function $f(m)$ is an appropriate model for the rest of the applicable figures in the table may be seen by noting that, given $n\%$ on the right column of Table 2, rounding the value of $f(\frac{n}{100})$ to the nearest ten-thousandth yields the corresponding decimal on the left column. To model the applicable figures for 2019, say, we use the same functional form with the values $j=0.0208$, $k=0.0311$, $\ell=0.0415$, $a=0.0654$, $b=0.0836$, $c=0.0986$ instead \cite[p 9]{IF8962}.} by defining $f(m)$ by
\[
f(m)=\begin{cases}
j, & 1\leq m < 1.33,\\
k+(\ell-k)\tfrac{m-1.33}{1.5-1.33}, & 1.33\leq m<1.5,\\
\ell+(a-\ell)\tfrac{m-1.5}{2-1.5}, & 1.5\leq m< 2,\\
a+(b-a)\tfrac{m-2}{2.5-2}, & 2\leq m< 2.5,\\
b+(c-b)\tfrac{m-2.5}{3-2.5}, & 2.5\leq m < 3,\\
c, & 3\leq m\leq 4,
\end{cases}
\]
where
\[
(j,k,\ell, a, b, c)=(0.0201, 0.0302, 0.0403, 0.0634, 0.0810, 0.0956).
\]

As Figure \ref{fig:ApplicableFiguresGraph} shows, our function has a discontinuity at $m=1.33$ and values in the interval $(0,0.1)$. Other tax years call for different values for $f(m)$. For all tax years, however, this function is defined similarly to the above example, is monotone increasing, and, by the grace of Congress, possesses right continuity on $[1,4]$.

Once we have the applicable figure, the remaining ingredients for the computation of the premium tax credit are readily obtained. The household's \emph{expected contribution}, found by multiplying the applicable figure $f(m)$ by household income $M$, is what the government expects the household to be able to affordably contribute towards health insurance. The expected contribution $f(m)\cdot M$ is compared with the premium $P$, the sum of the unsubsidized or ``sticker price'' costs of benchmark annual\footnote{If enrollees, plans, or premiums change from month to month, then monthly premiums must be recorded individually. And, in certain circumstances, costs are split between multiple tax returns. We assume, for simplicity, that neither of these occur.} health insurance premiums for the household members. The government is willing to ``pick up the tab,'' that is, pay all of the cost of the benchmark premiums which is not covered by the expected contribution. Thus, the government can pay the remaining amount, $P-f(m)\cdot M$. More precisely, as the government's contribution is never negative, it pays up to $\max\left(0, P-f(m)\cdot M\right)$.

\emph{Simplified Example.} Say an unmarried 60-year-old nonsmoker forms a household of $1$ person in Dutchess County, in the state of New York, in 2018. Assume their benchmark premium is $\$500$ per month, $F=\ $\$12,000, $M=\ $\$48,000, and $f(4)=0.09$. Then, their benchmark annual health insurance premium is
\[
P=12\cdot \text{\$500} = \text{\$6,000},
\]
and their expected contribution is
\[
f\left(m\right)\cdot M = 0.09\cdot \text{\$48,000} = \text{\$4,320}.
\]
If they buys the benchmark insurance, then the government pays the rest, which is
\[
P-f(m)\cdot M=\text{\$6,000}-\text{\$4,320}=\text{\$1,680}.
\]
Thus, \$1,680 is the amount of Obamacare's premium tax credit for the household.

We recall that the premium tax credit can be taken in advance. After doing this, the remaining balance on the benchmark\footnote{The benchmark insurance for a household member is the ``second lowest cost silver plan'' on the government exchange for the household's county of residence, and depends on the age and smoking habits of the enrollee, in addition to the county and tax year.} premiums can be paid with the expected contribution. As the government considers the expected contribution affordable, we have seen how, according to government definitions, Obamacare makes health coverage affordable\footnote{In particular, if all available employer-sponsored health insurance plans require an employee contribution that exceeds the household's expected contribution $f(m)\cdot M$, so they are ``unaffordable,'' then the household may generally purchase suitable insurance on the exchange, and receive a premium tax credit, to get affordable qualified insurance.} when $M$ satisfies $F\leq M\leq 4F$.

A household need not buy the benchmark insurance to receive a premium tax credit, however, and can purchase other qualified insurance from the exchange. Let $Q$ denote the sum of the unsubsidized annual costs of such qualified insurance, for health plans actually purchased for the household members. The government is still willing to contribute $\max\left(0,P-f(m)\cdot M\right)$ or the full cost $Q$ of the chosen qualified insurance, if this is less, since the government cannot pay more than the full cost. Thus, in general, the premium tax credit is given by
\[
PTC=\min\left(Q,\max\left(0,P-f(m)\cdot M\right)\right).
\]
Enrollees receive a form\footnote{Forms 1095-A list annual values for $Q$ and $P$ on lines 33A and 33B, respectively, for purchasers \cite{F8962}. They list annual values of advance payments of premium tax credits, $APTC$, on line 33C.} giving the values they should use for $Q$ and $P$ when filing taxes. To recap, the instructions for claiming the premium tax credit give the value of $f(m)$ from tables once $M$ is known. Then, using $Q, P$ read off a form, the premium tax credit $PTC$ is calculated by the above formula. With such a precise process available, why was my Uber driver unable to find his premium tax credit?

\section{The Problem}

We have yet to define or calculate household income $M$. There's a reason for that. For households with income from \emph{self-employment}---an independent contractor, a private tutor, and a driver associated with a ridesharing app are all likely to be considered self-employed---the value of $M$ can be tricky to find, particularly if the household is eligible for a premium tax credit. In 2014, self-employed workers were ``almost three times more likely'' than other workers to obtain health insurance from the government exchanges created by Obamacare, according to the Treasury \cite{LMT}, so self-employed households form a sizable proportion of beneficiary households. Thus, we are motivated to address any computational issues they may face; such issues could potentially affect a large number of people.

Pinning down the household income $M$ for self-employed households like that of my Uber driver requires detective work. This is because self-employed households are eligible for a tax deduction $D$ involving health insurance costs, which may be difficult to determine when those costs are being shared with the government. Rather than discussing household income\footnote{The Instructions for Form 8962 \cite[p 6]{IF8962} say that household income in general is the sum of the modified adjusted gross incomes of the household members. Each modified adjusted gross income is the sum of the adjusted gross income (AGI) on the corresponding tax return and certain tax-exempt income. The AGI itself is obtained from total or gross income by subtracting certain adjustments, called ``above the line'' deductions.} in general, we just discuss it for households whose sole income source is self-employment in a single business. Let us denote the \emph{earned income}\footnote{Earned income is defined to be the net profit from a business minus the self-employment tax deduction---corresponding to the half of Social Security and Medicare taxes normally paid by employers---and tax-deductible self-employed retirement plan contributions.} generated from this activity by $I$. If the household's health insurance is all purchased on the exchange by this business, then some nonnegative amount $D$ of that cost can be deducted\footnote{We assume, for simplicity, that earned income $I$ is at least as large as $Q$. Otherwise, we must also require that the self-employed health insurance deduction $D$ satisfy $D\leq I$.} from $I$, so taxes are only paid on the amount $I-D$. If there are no other sources of ``above the line\footnote{The origin of the name is that on Form 1040 \cite{F1040} for 2017 and many prior years, AGI appears at the bottom of the front page, with a line underneath. The deductions needed to compute AGI---the ``above the line'' deductions---appear above that line, while all other deductions do not.}'' deductions besides $D$ and the ones used to compute $I$, then the household income $M$ is given by
\[
M=I-D.
\]
We say that the household has a ``simple'' tax return in this case, since it only has one income source and one above the line deduction besides those used to find $I$.

For simple tax returns, the premium tax credit can be determined from $M=I-D$, and hence from $D$, but what range of values can $D$ have? We introduce two legal constraints on $D$. The first is that the government doesn't permit more to be deducted than the household was billed for during the year, so $D\leq Q$. If advance payments of premium tax credits were sent, and we denote the total amount sent by $APTC$, then the balance billed for was $Q-APTC$. So, $D\leq Q-APTC$ in this case, but for simplicity we take $APTC=\$0$ for now. The second constraint, which we call the ``no double-dipping rule,'' is that the government doesn't permit more in deductions and credits than was possible to pay. Without this restriction, an enterprising person might buy health insurance at a negative effective cost, presumably contrary to the taxpayers' wishes. In our case, we write the rule as
\[
D+PTC(D)\leq Q,
\]
where $PTC(D)$ is the amount of the premium tax credit for a household with income $M=I-D$. We can make this function explicit by replacing $m$ with $M/F$ and $M$ with $I-D$ in the equation for $PTC$ from the previous section, giving
\[
PTC(D)=\min\left(Q,\max\left(0,P-f\left(\frac{I-D}{F}\right)\cdot \left(I-D\right)\right)\right).
\]
Since the second constraint implies the first for simple tax returns with $APTC=\$0$, we ignore the first for now, and may refer to the second simply as ``the constraint.''

We now come to my Uber driver's dilemma. To find his premium tax credit, he must know $D$. But $D\leq Q-PTC(D)$, so he must know his premium tax credit to find out how large $D$ can be. But the premium tax credit is what he wanted to find in the first place! So, there is a ``circular relationship'' in the United States tax code between the premium tax credit and the self-employed health insurance deduction \cite{IRS}. This means that the Internal Revenue Service (IRS) has the following problem:

\begin{problem}
What is a procedure, computable by hand in a reasonable time, that finds the appropriate health insurance deduction $D$ for any self-employed household eligible for Obamacare's premium tax credit?
\end{problem}

What does ``appropriate'' mean? The appropriate choice of $D$ is the nonnegative value which maximizes the tax benefit for the household. The tax benefit is the sum of the tax credit\footnote{In case $m=\tfrac{I-D}{F}$ falls outside of $[1,4]$ during our analysis, we take $PTC(\cdot)$, as a function of $D$, to have the value $PTC(D)$ given by the above formula if $m$ lies in $[1,4]$, and we take $PTC(D)=0$ if $m>4$ or $m<1$, for now. However, there can be exceptions to this when $m<1$.} $PTC(D)$ and the amount of taxes saved by reducing income by the deduction $D$. If the tax function $T(\cdot)$ assigns, to a given income in dollars, the federal income tax on that value for a household, ignoring the tax credit, then the amount of taxes saved by reducing taxable income from $I$ to $I-D$ is $T(I)-T(I-D)$. With these definitions, the tax benefit is
\[
PTC(D)+\left(T(I)-T(I-D)\right).
\]
As $T(I)$ is independent of $D$, the optimal solution is unaffected by dropping it from the problem formulation. Thus, the appropriate deduction is the value of $D$ solving
\[
\max_{D+PTC(D)\leq Q}PTC(D)-T(I-D).
\]
Essentially, we want to maximize the tax credit while minimizing the tax, subject to the constraint. As an increase in $D$ causes a decrease in $I-D$, whence an increase in $PTC(D)-T(I-D)$, the largest value of the latter occurs for the largest $D$ satisfying the constraint. That is, the appropriate nonnegative value of $D$ is
\[
\max\left(\{D:D+PTC(D)\leq Q\}\right),
\]
provided $M=I-D\geq F$ for this value, so the household is eligible for the credit.

What does ``computable by hand in a reasonable time'' mean, above? Practically, it means the IRS can put it into its tax guidance. Informally, this means the IRS does not consider it overly onerous to require of a typical taxpayer with access to its instructions, even if removed from modern computing. For example, if we try all possible whole dollar values for $D$ that satisfy the constraint, then, to the nearest dollar, some value will yield the maximum tax benefit, and thus will give the appropriate $D$. But the IRS would likely consider having to try every possible constrained value $D$ to be an overly burdensome computational task for an American unable to access a computer or smartphone. Thus, although guaranteed to succeed, this is not a procedure that any household can ``compute by hand in a reasonable time.'' On the other hand, the maximization problem for $D$ can eventually be converted to an algebraic equation in $D$ for each taxpayer. This is because $f(m)\cdot M$ is a piecewise-quadratic function of $D$ and the constraint  is an analyzable inequality. However, the IRS would probably find it unreasonable to require an American removed from the internet to discover the necessary algebra and numerical computation of square roots. The task is to create an algorithm or procedure which can be implemented in a reasonable number of steps that just involve addition, subtraction, multiplication, division, and rounding. While this may be possible for the above-mentioned algebraic equations, it would likely use many specific details about the function $f(m)$, so the guidance would have to be rewritten each year using the new year's function. It would be preferable for the IRS to derive a dependable procedure that is independent of the tax year.

If we can solve the above problem for simple tax returns with $APTC=\$0$, then we can check whether our solution, appropriately generalized, handles simple tax returns with positive values of $APTC$, and general tax returns. We turn now to current IRS guidance for taxpayers who qualify for both a self-employed health insurance deduction $D$ and a premium tax credit $PTC(D)$. This guidance can be viewed as an attempted solution of the above problem.

\section{The IRS Fixed Point Iteration}

Current IRS guidance offers self-employed Obamacare beneficiaries two methods for determining allowable values of their self-employed health insurance deduction $D$ \cite[pp 62--65]{IRS974}. The second, ``simplified calculation method'' is a truncation of the first, ``iterative calculation method,'' so we focus primarily on motivating and analyzing the IRS iterative method here.

To motivate the IRS iteration for finding the appropriate self-employed health insurance deduction $D$, we ask what equations $D$ might satisfy, in plausible scenarios. Certainly, if we write the premium tax credit as a variable, $C$, we have the equation
\[
C=PTC(D)
\]
by definition. In addition, we might hope that the appropriate $D$, the largest nonnegative value satisfying $D+PTC(D)\leq Q$, attains the equality $D+PTC(D)=Q$. This means each dollar earmarked for health insurance leads to a dollar of tax credit or a dollar of insurance deductions, a plausible property for the $D$ giving the greatest tax benefit to provide. Subtracting $PTC(D)$ in this equation, we get
\[
D=Q-PTC(D),
\]
so we arrive at a system of two equations in two unknowns given by
\[
(C,D)=(PTC(D),Q-PTC(D)).
\]

Rather than try to solve the above system of two equations in two unknowns algebraically, the IRS uses the two earlier equations to define a fixed point iteration. Given $(C_{n}, D_{n})$ for some integer $n\geq 1$, we can define $(C_{n+1}, D_{n+1})$ by
\[
(C_{n+1},D_{n+1})=(PTC(D_{n}),Q-PTC(D_{n})).
\]
This indeed defines a fixed point iteration, for if we define $G$ on $\mathbb{R}^{2}$ by
\[
G(x,y)=\left(PTC(y),Q-PTC(y)\right),
\]
then the above equation for $(C_{n+1}, D_{n+1})$ becomes
\[
(C_{n+1},D_{n+1})=G(C_{n},D_{n}).
\]
As we shall see, this iteration seeks a limiting point $X=(C,D)$ such that $X=G(X)$; such a point is said to be fixed by $G$.

For an eligible household with a simple tax return and no advance premium tax credit, the IRS iterative method generally begins with the point $(C_{1}, D_{1})$ given by $C_{1}=$ \$0, $D_{1}=Q$ and defines points $(C_{n},D_{n})$ sequentially by the above equation $(C_{n+1}, D_{n+1})=G(C_{n},D_{n})$. This is the case, for example, if $I-Q\geq F$. For our model, if we define the sequence $\{(C_{n},D_{n})\}_{n=1}^{\infty}$ in this way, then we can ask about convergence. If $(C_{n}, D_{n})\to (C,D)$ as $n\to\infty$, then, taking limits on both sides of the equation $(C_{n+1}, D_{n+1})=G(C_{n}, D_{n})$, we can prove\footnote{\emph{Proof.} We first observe, inductively, that $D_{2}\leq D_{4}\leq D_{6}\leq\cdots$ when $(C_{1}, D_{1})=(\$0,Q)$. Then, we apply the left continuity of the function $PTC(\cdot)$, inherited from the right continuity of $f(\cdot)$, to the sequence $\{D_{2n}\}_{n=1}^{\infty}$. Our hypothesis that $(C_{n}, D_{n})\to (C,D)$ as $n\to\infty$ then yields $PTC(D)=\lim_{n\to\infty}PTC(D_{2n})=C$.} that
\[
(C,D)=\left(PTC(D),Q-PTC(D)\right)
\]
holds. From this, we see that $(C,D)=G(C,D)$, so $(C,D)$ is a fixed point of $G$, whence $D+PTC(D)=Q$. This shows $D$ satisfies the constraint. Moreover, as $D+PTC(D)$ is a strictly increasing function of $D$ when $m=\tfrac{I-D}{F}\geq 1$, no larger value of $D$ satisfies the constraint. Thus, $D=\max\left(\{D:D+PTC(D)\leq Q\}\right)$, that is, $\lim_{n\to\infty}D_{n}$ is the appropriate value of $D$.

Having motivated the IRS fixed point iteration with the usual notion of convergence, we point out that the IRS uses its own test to determine convergence. First, let us say that a sequence $\{(C_{n}, D_{n})\}_{n=1}^{\infty}$ \emph{converges in the IRS sense} if and only if, when rounding to the nearest penny after each intermediate calculation, there exists a positive integer $N$ such that
\[
\|(C_{k}, D_{k})-(C_{n}, D_{n})\|_{\infty}<\varepsilon_{0}
\]
for all integers $k,n\geq N$, with $\varepsilon_{0}=\$1$. The above norm is defined by $\|(x,y)\|_{\infty}=\max(|x|,|y|)$. It is straightforward to show, for the sequences $\{(C_{n}, D_{n})\}_{n=1}^{\infty}$ defined in the preceding paragraph, that the ``iterative calculation method'' amounts to taking $D$ to be the $y$-coordinate of $(C_{n_0}, D_{n_0})$ after appropriate rounding, where $n_0$ is the smallest value of $N$ that satisfies our definition of IRS convergence \cite[p 63]{IRS974}. The actual text of IRS guidance asks taxpayers to ``not use the iterative calculation method'' if $\|(C_{n+1}, D_{n+1})-(C_{n}, D_{n})\|_{\infty}\geq\varepsilon_{0}$ for all $n\geq 1$. The oscillatory nature of the sequences makes it possible to prove that this holds if and only if $\{(C_{n},D_{n})\}_{n=1}^{\infty}$ fails to converge in the IRS sense \cite[p 64]{IRS974}. Thus, for the sequences coming from its iteration, the IRS provides a simple, accurate ``do not use'' test\footnote{We note that this test for divergence is equivalent to using the well-known Cauchy criterion for this with $\varepsilon_{0}=\$1$.} for divergence.

IRS guidance also offers its simplified calculation method, which amounts to asking beneficiaries to take $D_{2}$ as their health insurance deduction and, hence, $C_{3}$ as their premium tax credit. When $\{(C_{n},D_{n})\}_{n=1}^{\infty}$ fails to converge in the IRS sense, so that we cannot use the IRS fixed point iteration, these are the values that IRS guidance currently arrives at for $D$ and $C$. The best of the tax software may extend the simplified procedure, so taxpayers take at most
\[
D_{0}:=\liminf_{n\to\infty}D_{n}
\]
as their deduction, and hence $PTC(D_{0})$ as their premium tax credit. When we don't have convergence in the IRS sense, however, $D_{0}$ is generally smaller than the appropriate value, and in many cases $PTC(D_{0})$ yields a premium tax credit of \$0. This is apparently the cause of my Uber driver's difficulty; we emphasize that these inappropriate values are what tax software and government calculators give now.

Unsurprisingly, the IRS says that self-employed taxpayers ``may have difficulty'' computing their premium tax credit, according to the IRS document which introduced the fixed point iteration \cite{IRS}. However, IRS guidance says that ``any'' computation method may be used to find the appropriate deduction, provided it respects the constraint and the separate rules for the deduction and credit \cite[p 64]{IRS974}. The below example proves that neither method given by the IRS always works to compute appropriate values, so the fact that any valid method may be used gives us a fresh opportunity to solve the IRS problem---and my Uber driver's dilemma.

\emph{Example.} Say we are considering the 2018 tax year. Then the example of an applicable figure function $f(m)$ given previously is the appropriate one to use. In particular, we can use it to calculate values of the extended tax credit function
\[
PTC(D)=\begin{cases}
\min\left(Q,\max\left(0,P-f(\frac{I-D}{F})\cdot (I-D)\right)\right), & 1\leq\frac{I-D}{F}\leq 4,\\
0, & \frac{I-D}{F}>4\text{ or }\frac{I-D}{F}<1.
\end{cases}
\]
Suppose we have a household in Brooklyn, New York, consisting of one individual and one dependent child who is less than 26 years old. The household's relevant federal poverty line is $F=\text{\$16,240}$ \cite[p 7]{IPTC18}. Looking up benchmark prices for the county, Kings County, we find that the unsubsidized cost of benchmark health insurance premiums for the household is \$865.81 per month or, rounding to the nearest dollar, $P=\text{\$10,390}$ annually \cite{NYH}. Suppose that the household, altogether, has earned self-employment income from a single business which amounts to $I=\text{\$71,150}$, and take $Q=P$. Following the IRS fixed point iteration, and rounding to the nearest dollar in intermediate steps for simplicity, we obtain
\[
(C_{1}, D_{1})=(\$0, \text{\$10,390})
\]
and
\[
C_{2}=\text{\$10,390}-0.0956\cdot\text{\$60,760},
\]
as $\text{\$71,150}-\text{\$10,390}=\text{\$60,760}$.
Hence, after rounding, $C_{2}$ is \$4,581. Thus,
\[
D_{2}=\text{\$10,390}-\text{\$4,581}=\text{\$5,809}.
\]
In turn, this makes $I-D_{2}=\text{\$65,341}>4F=\text{\$64,960}$, so by our above formula for $PTC(D)$, we get
\[
C_{3}=\$0.
\]
Unfortunately, this yields
\[
D_{3}=\text{\$10,390},
\]
putting us back where we started. Hence, the sequence doesn't converge in the IRS sense. On the other hand, if we follow the simplified calculation method from the IRS, we arrive at a deduction of $D_{2}=$\$5,809 and a premium tax credit of $C_{3}=\$0$. This is even worse than not claiming the premium tax credit at all, and letting $D=$\$10,390. It turns out that the \$0 value for the premium tax credit is not appropriate, as we shall see. If we progressively narrow our search for the deduction $D$ by performing repeated bisections, for example, then we can do better.

\section{The Bisection Method}

We now propose a bisection procedure, and prove that it always gives the appropriate self-employed health insurance deduction $D$ for simple tax returns. The proof works because, although there may, in general, be discontinuities in the underlying structures that affect potential computations, we have monotonicity and left continuity in the function $PTC(D)$, the latter inherited from the right continuity of $f(m)$ through $m=\tfrac{I-D}{F}$. We first motivate the use of bisection by adapting the well-known proof of the Intermediate Value Theorem by bisection to left continuous, monotone increasing functions.

\begin{theorem}
Let $g$ be a real-valued, monotone increasing, left continuous function on an interval $[a,b]$, and let a real number $k$ be given. If $g(a)\leq k$, then there exists a real number $c$ in $[a,b]$ such that $g(c)\leq k$ and $g(d)>k$ for all $d>c$ in $[a,b]$.
\end{theorem}

\begin{proof}
If $g(b)\leq k$ then, as there is no $d>b$ in $[a,b]$, we can set $c=b$. Otherwise, if $g(b)>k$, we denote the midpoint of $[a,b]$ by $c_{1}=\tfrac{a+b}{2}$. If $g(c_{1})>k$, we set $a_{1}=a$, $b_{1}=c_{1}$, and reduce our search to $[a_{1}, b_{1}]$ as, by monotonicity, $g(d)>k$ for all $d\geq c_{1}$ in $[a,b]$. If $g(c_{1})\leq k$, we set $a_{1}=c_{1}$, $b_{1}=b$, and again reduce our search to $[a_{1}, b_{1}]$. Having defined $[a_{n}, b_{n}]$ for some $n\geq 1$ such that $g(a_{n})\leq k<g(b_{n})$, we let $c_{n+1}=\tfrac{a_{n}+b_{n}}{2}$. If $g(c_{n+1})>k$, we set $a_{n+1}=a_{n}$, $b_{n+1}=c_{n+1}$, whereas if $g(c_{n+1})\leq k$ we set $a_{n+1}=c_{n+1}$, $b_{n+1}=b_{n}$. In either case, we have $g(a_{n+1})\leq k <g(b_{n+1})$. Having defined $\{[a_{n}, b_{n}]\}_{n=1}^{\infty}$ recursively as above, it follows that a unique $c$ lies in all of the intervals $[a_{n}, b_{n}]$, as $[a_{n+1}, b_{n+1}]$ is a subset of $[a_{n}, b_{n}]$ and $[a_{n}, b_{n}]$ has length $\tfrac{b-a}{2^{n}}$ for all $n\geq 1$. Note that $c$ is the limit of the increasing sequence $\{a_{n}\}_{n=1}^{\infty}$. Since $g(a_{n})\leq k$ for all $n\geq 1$, and $a_{n}\to c$ as $n\to\infty$, it follows that $g(c)\leq k$ by left continuity of $g$, as desired. Given $d>c$ in $[a,b]$, for $n$ sufficiently large we have $\tfrac{b-a}{2^{n}}<d-c$. For such $n$, we have $c\leq b_{n}<d$. As $g(b_{n})>k$ and $g$ is increasing, it follows that $g(d)>k$. Thus, as $d>c$ in $[a,b]$ was arbitrary, $c$ is as desired.
\end{proof}

We can apply the proof of this theorem to justify the appropriateness of using bisection to calculate $D$ as in the following corollary. This corollary is for simple tax returns with no advance premium tax credit. For such returns, if $I<F$, then $PTC=\$0$ automatically, so the premium tax credit cannot be taken and the optimal deduction\footnote{If $I<Q$ is possible, then the IRS requires us to further constrain $D$ by $D\leq I$. So, the highest possible deduction is actually $D=\min(I,Q)$ when $I<F$, but we have assumed $I\geq Q$ above for simplicity and will continue to do so below.} is $D=Q$. The following corollary handles the remaining case, where $I\geq F$. We recall that $D$ must lie in $[0,Q]$ and that we must have $D\leq I-F$, so that $M=I-D\geq F$, to ensure eligibility for the premium tax credit in our current setup. The tax benefit when $D>I-F$, so that $PTC=\$0$ and $D=Q$, can be considered separately. However, this poses no computational issue and the benefit is smaller than taking the premium tax credit in practical situations. That is why we only consider $D\leq I-F$ here.

\begin{corollary}
Suppose $F, P, Q>0$, $I\geq F$ are given real numbers, and $f$ is a positive, monotone increasing, right continuous function on $[1,4]$. Define $PTC(\cdot)$ on $(-\infty, I-F]$ by letting
\[
PTC(d)=\begin{cases}
\min\left(Q,\max\left(0,P-f(\frac{I-d}{F})\cdot (I-d)\right)\right), & 1\leq\frac{I-d}{F}\leq 4,\\
0, & \frac{I-d}{F}>4.
\end{cases}
\]
Then, if $PTC(0)\leq Q$, there is $D$ in $\left[0, \min\left(Q, I-F\right)\right]$ such that $D+PTC(D)\leq Q$ and $d+PTC(d)>Q$ for all $d>D$ in $\left[0, \min\left(Q, I-F\right)\right]$. Such $D$ may be computed by bisection. Letting $a_{0}=0$ and $b_{0}=\min(Q, I-F)$, this means that if $b_{0}+PTC(b_{0})\leq Q$ then $D=b_{0}$ and if $b_{0}+PTC(b_{0})>Q$ then we can proceed as follows. For each integer $n\geq 0$, having obtained $a_{n}$ and $b_{n}$ with $a_{n}+PTC(a_{n})\leq Q < b_{n}+PTC(b_{n})$, let $c_{n+1}=\tfrac{a_{n}+b_{n}}{2}$. If $c_{n+1}+PTC(c_{n+1})\leq Q$, then let $a_{n+1}=c_{n+1}$, $b_{n+1}=b_{n}$. Otherwise, if $c_{n+1}+PTC(c_{n+1})>Q$, let $a_{n+1}=a_{n}$, $b_{n+1}=c_{n+1}$. The increasing sequence $\{a_{n}\}_{n=0}^{\infty}$ defined by this procedure converges to the number
\[
D=\max\left(\left\{d\text{ \emph{in} }[0,\min(Q,I-F)]:d+PTC(d)\leq Q\right\}\right)
\]
with the property above.
\end{corollary}

\begin{proof}
Let $a=0$, $b=\min\left(Q, I-F\right)$, and define $g$ on $[a,b]$ by $g(d)=d+PTC(d)$. By the preceding theorem, if $PTC(a)\leq Q$, then $D$ in $[a,b]$ with the desired property exists. The proof of that theorem justifies the rest of the assertion.\end{proof}

If the function $g$ in the above proof is such that $g(0)>Q$, then no $D$ with the desired property exists, but this is not a problem in our application as $PTC=\$0$ in this case and we can take $D=Q$. If $g(0)\leq Q$ and $g(b)\leq Q$ too, where $b=\min(Q,I-F)$, then the appropriate $D$ must be $D=b$. In all remaining cases, we can perform repeated bisections of $[0, b]$ as described above.

\emph{Example.} We can perform the bisection procedure on the example from the preceding section. After more than a dozen bisections, rounding to the nearest dollar after each step, we find that the appropriate deduction is $D=\$$6,208. From this, we find that the appropriate premium tax credit for the household is $PTC(D)=\$$4,182, substantially more than the \$0 it would receive by the simplified calculation method. Given $D$, and without performing the bisection, it is readily checked that this value cannot be improved because $D+\text{\$4,182}=Q$. By means of such checking, the IRS apparently verifies the correctness of tax returns prepared using values of $D$ found by methods outside of its guidance, such as bisection.

Similarly to the above example, when $PTC(D)=P-f(m)\cdot M$ and $Q\geq f(4)\cdot 4F$, there is generally an interval of incomes $I\leq 4F+f(4)\cdot 4F$ for which the IRS fixed point iteration breaks down. In this case, the simplified calculation method gives $PTC=\$0$, yet an appropriate deduction $D$ can be found by bisection yielding substantial premium tax credits, often worth thousands of dollars.

The bisection method also offers improvement over IRS guidance near $m=1.33$, as the discontinuity in $f(m)$ at $1.33$ again prevents IRS convergence nearby. In fact, it is possible for the equation $D+PTC(D)=Q$ to have no solution for some interval of incomes $I$, due to the discontinuity at $m=\tfrac{I-D}{F}=1.33$. For such $I$, there is a value of $D$ such that $D+PTC(D)<Q$ yet $d+PTC(d)>Q$ for $d>D$ due to a discontinuous jump in $PTC(d)$ as $d$ approaches $D$ from the right. Thus, this $D$ is the appropriate value for the deduction, yet $D+PTC(D)<Q$. In our model of the IRS fixed point iteration, if $\{(C_{n}, D_{n})\}_{n=1}^{\infty}$ converges to $(C, D)$, then $D+PTC(D)=Q$ necessarily follows. For this reason, in this interval of incomes $I$, it is impossible for the IRS fixed point iteration to converge; it can be shown that this is true no matter how we select the initial point $(C_{1}, D_{1})$.

%

\section{The Advance Premium Tax Credit}

Having developed the bisection procedure for simple tax returns with no advance premium tax credit, we consider the situation where advance payments of the premium tax credit are sent. Let $APTC$ denote the amount of advanced payments in this case, and let us treat \$0 as $0$. Then, $D$ must lie in $[0,Q-APTC]$, but otherwise the computation of the premium tax credit may begin as before, starting from the interval $[0,Q-APTC]$. It is usually the case that, if a taxpayer with $APTC>0$ receives the advance payments based on the expectation that $m=\tfrac{I-D}{F}\geq 1$ will occur, then the taxpayer may still claim the premium tax credit even if $m<1$ actually occurs \cite[p 8]{IF8962}. In such cases of unstable self-employment income, we can use the equation for $PTC(D)$ that is given when $1\leq m\leq 4$, using the actual values of $I$ and $D$ but using $f(1)$ as the applicable figure, as mentioned previously. We obtain the following corollary after taking $f(m)=f(1)$ for $m$ in $[0,1]$. This is the result with the widest practical application in this paper.

\begin{corollary}
Suppose $F, P, Q>0$, $I\geq Q$, and $0<APTC\leq Q$ are given real numbers, and $f$ is a positive, monotone increasing, right continuous function on $[0,4]$. Define $PTC(\cdot)$ on $(-\infty, I]$ by letting
\[
PTC(d)=\begin{cases}
\min\left(Q,\max\left(0,P-f(\frac{I-d}{F})\cdot (I-d)\right)\right), & 0\leq\frac{I-d}{F}\leq 4,\\
0, & \frac{I-d}{F}>4.
\end{cases}
\]
Then, if $PTC(0)\leq Q$, there is $D$ in $[0, Q-APTC]$ such that $D+PTC(D)\leq Q$ and $d+PTC(d)>Q$ for all $d>D$ in $[0, Q-APTC]$. Moreover, the number
\[
D=\max\left(\left\{d\text{ \emph{in} }[0,Q-APTC]:d+PTC(d)\leq Q\right\}\right)
\]
with the property above can be computed by bisection, starting from $a_{0}=0$ and $b_{0}=Q-APTC$, as in the previous corollary.
\end{corollary}

Lastly, when $APTC>0$, the value $APTC$ of the advance payments must be reconciled with the appropriate premium tax credit $PTC(D)$ when taxes are filed. If $PTC(D)\geq APTC$, the taxpayer receives an additional $PTC(D)-APTC$ when filing taxes, so the total amount of $PTC(D)$ is ultimately received. However, if $APTC>PTC(D)$, then the taxpayer must repay some or all of the advance payments which were received in excess of the appropriate amount. Precisely, according to IRS guidance, the taxpayer repays
\[
\min\left(APTC-PTC(D),R\left(\frac{I-D}{F}\right)\right),
\]
where, for nonnegative $m$, the \emph{repayment limitation} $R(m)$ is of the form
\[
R(m)=\begin{cases}
r, & 0\leq m < 2,\\
s, & 2\leq m<3,\\
t, & 3\leq m < 4,\\
\infty, & 4\leq m
\end{cases}
\]
with $r\leq s\leq t$ \cite[p 16]{IF8962}. For example,
\[
(r,s,t)=(300, 800, \text{1,325})
\]
for tax year 2019 \cite[p 16]{IF8962} if the filing status of the taxpayer is ``single,'' and the values of $r, s, t$ are doubled for any other filing status. For comparison, the values of $r, s, t$ for single taxpayers in 2018 were \$\text{300}, \$\text{775}, \$\text{1,300}, respectively \cite[p 16]{IPTC18}. Independent of the tax year, $R(m)$ is increasing, positive, and right continuous.

The total tax benefit from Obamcare when $APTC>PTC(D)$, which is the amount received minus the amount repaid, is
\[
B(D)=APTC-\min\left(APTC-PTC(D),R\left(\frac{I-D}{F}\right)\right),
\]
and this need not equal $PTC(D)$. However, there is no need to modify the bisection in the previous corollary to take this new formula into account. To see this, first note that the total Obamacare tax benefit $B(D)$, for the $D$ found from the bisection, is at least $PTC(D)$ and less than $APTC$ in the case we are considering, where $APTC>PTC(D)$. As $B(D)<APTC$, and $D\leq Q-APTC$, the ``no double-dipping rule'' is automatically satisfied in this case, as summing gives $D+B(D)<Q$. Moreover, as $B(D)\geq PTC(D)$, all values $d>D$ in $[0,Q-APTC]$ persist in being inappropriate when we use $B(D)$ as the effective credit. Indeed, $d+PTC(d)>Q$ implies that $d+B(d)>Q$ in this case. Thus, the appropriate value of $D$ in $[0, Q-APTC]$ found in the previous corollary is still optimal when taking into account the total Obamacare tax benefit, even when $APTC>0$.

\section{Further Questions}

We have seen that we can perform bisection for models of simple tax returns, whether or not there are advance payments of the premium tax credit. What about models of general tax returns? We briefly sketch the main ideas. Additional sources of income, whether tax-exempt or not, do cause translations in $M$, but as long as earned income from a single self-employed business is at least $Q$, the previous considerations still apply to $M=I-D$, where now $I$ represents the sum of all income sources relevant to computing household income. Thus, this increase in generality causes no difficulty. If there are further above the line deductions which are not from a short list of exceptions \cite[p 65]{IRS974}, or insurance deductions coming from dental or other plans that are unrelated to the premium tax credit, then they again cause a translation in $M$. Effectively, the previous considerations apply to $M=I-d_{0}-D$, where we let $d_{0}$ represent the sum of these additional deductions.
From this, we see generalization again causes no difficulty. Finally, for the tax code as it stands, there are above the line deductions that can be altered when the value of $D$ is adjusted. For example, the student loan interest deduction for 2019 is normally \$2,500 or the total interest paid on student loans during the tax year, whichever is less \cite[p 37]{SL19}. Let us denote the smaller of these two amounts by $k$. When the taxpayer's modified adjusted gross income (MAGI) is between \$70,000 and \$85,000, the deduction becomes the amount
\[
SL=k\cdot\frac{\$85,000-MAGI}{\$\text{15,000}}.
\]
Roughly speaking, the above value for MAGI is found using an adjusted gross income (AGI) computed with a value of \$0 for the student loan interest deduction. The above formula shows, for example, that the student loan interest deduction ``phases out'' from $k$ to $\$0$, with a slope of $-k/\$\text{15,000}$, as income ranges from \$70,000 to \$85,000. If the self-employed health insurance deduction $D$ is increased, that could ``undo'' some portion of the phasing out in certain circumstances. This can cause the student loan interest deduction, and other above the line deductions like it, to get caught up in the circular relationship between $D$ and the premium tax credit. However, because all of the current exceptional above the line deductions of this type are deductions that phase out like this one, increasing $D$ causes each of them, when successively calculated, to be monotone increasing as functions of $D$, as can be verified by examining them individually \cite[p 65]{IRS974}. For this reason, $PTC(D)$ persists in being monotone in $D$ for general tax returns, and for this reason bisection can currently be appropriately adapted for models of general returns.

What is perhaps more mathematically interesting than our ability to adapt a well-known algorithm is the task of explaining in detail when and how the IRS procedures break down. As we have seen, this is equivalent to the question of when the IRS fixed point iteration fails to converge in the IRS sense. The explanation, as we have emphasized above, seems to lie in the discontinuities that the function $PTC(D)$ possesses in general. However, because the precise way that the points of the iteration are jostled about is also influenced by the slopes involved, near $m=1.33$, we have not computed the precise intervals of IRS divergence in general. In particular, this problem is not fully resolved, and may be of interest, when $APTC>0$.

Another interesting question is that of investigating whether more sophisticated algorithms than bisection might be needed if tax laws were changed. The bisection procedure, as the proof of our theorem suggests, relies on a net monotonicity effect from the relevant above the line deductions when it is adapted to models of general tax returns.

A deduction which ``phases in'' may lead to a loss of monotonicity in $PTC(D)$ as a function of $D$ and fool the bisection method into finding a point which does not maximize tax benefits globally. There are currently no such above the line deductions of this type, but from 2005 to 2017 the domestic production activities deduction (DPAD) that might affect, for instance, someone who strikes oil in Texas, involved a continuous phase-in. The slope is sufficiently shallow, however, that when $D$ is decreased by $\varepsilon$, causing the other relevant above the line deductions to decrease by some $\varepsilon_{i}$, DPAD can phase in by at most $0.09(\varepsilon+\sum_{i}\varepsilon_{i})$. So, reversing this, an increase in $D$ by $\varepsilon$ still decreases AGI by at least $\varepsilon+\sum_{i}\varepsilon_{i}-0.09(\varepsilon+\sum_{i}\varepsilon_{i})$, or $0.91(\varepsilon+\sum_{i}\varepsilon_{i})$, leading to an increase in $PTC(D)$. Thus, monotonicity is preserved by this deduction.

What if a new deduction arises that phases in with a discontinuous jump upwards? In that case, the bisection method can indeed fail, and we should seek an alternative. A naive binned Newton method might succeed in many cases, but when $m$ is near 1.33, for example, and a discontinuous phase-in is suitably chosen, this could fail too. If the household's state has expanded Medicaid past the $(1.33)F$ threshold---and the Affordable Care Act prescribes Medicaid expansion\footnote{The Affordable Care Act provides funding to help expand Medicaid, which provides health coverage for many low income people, to all households with incomes below 138\% of the federal poverty line. However, the Medicaid program is run jointly by the federal government and the states. To ensure that the states cooperate in the Medicaid expansion, Obamacare provides for Medicaid funding to be reduced for states which refuse to expand the program along the lines it prescribes. The United States Supreme Court ruled \cite{USC} that the latter is ``coercive'' to the states, and struck down the reduced funding provision in 2012.} over and beyond this---we would expect the household to be eligible for Medicaid, so the household would likely be on Medicaid and hence unaffected by this discontinuity at $m=1.33$. However, if the state has not expanded Medicaid, this discontinuity might play a role or, worse, the household might fall into a ``Medicaid gap'' and be ineligible for both Medicaid and the premium tax credit, which requires at least an expectation that $M\geq F$ to get started. Further discontinuous numerical analysis might thus be needed to clarify the situation in states that have not expanded Medicaid, if a deduction with a discontinuous phase-in were to arise.

A policy-related question we might ask is how to modify the premium tax credit, as a function of household income, so that the IRS fixed point iteration always converges in the IRS sense. If $PTC(\cdot)$ is changed so that it is given by a monotone, continuous, piecewise-differentiable function, for example, with both of its one-sided derivatives of suitably small magnitude at each point, then we expect the IRS iteration to always converge. Such a change could be obtained by a modification that takes $f(1.33)=f(1)$, say, to make $f(m)$ continuous on $[1,4]$, and adjoins a suitable continuous phase-out of the premium tax credit for $M>4F$. While Congress may prefer to achieve objectives which are incompatible with this particular proposal, the fact that a circular relationship in the US tax code exists at all suggests that there are appropriate ways to involve mathematicians in realistic future policy choices. We hope this mathematical excursion, resolving the concern raised by my Uber driver, inspires some mathematicians to look for these and other ways they can use mathematics to help address citizens' concerns and, more generally, issues that affect the people they meet in daily life.

\emph{Acknowledgements.} This paper has been improved as a result of discussions with Xiaona Zhou at Brooklyn's City Tech about how to adapt the results presented here into an online self-employed premium tax credit calculator \cite{XZ}. Any feedback about the calculator should be addressed to \href{mailto:Xiaona.Zhou@mail.citytech.cuny.edu}{\nolinkurl{Xiaona.Zhou@mail.citytech.cuny.edu}} in order to be incorporated. Finally, I wish to thank the anonymous reviewers and my colleagues for their many helpful suggestions.

\end{document}